\documentclass[runningheads]{llncs}
\usepackage{graphicx}

\usepackage{cite}
\usepackage{amsmath,amssymb,amsfonts}
\usepackage{algorithmic}
\usepackage{graphicx}
\usepackage{textcomp}
\usepackage{xcolor}
\usepackage{multicol}

\usepackage{graphicx} 
\usepackage[vlined]{algorithm2e} 
\usepackage{url}
\usepackage{multirow}
\usepackage[normalem]{ulem}
\usepackage{tikz}
\usetikzlibrary{shapes.gates.logic.US,trees,positioning,arrows}
\usetikzlibrary{calc}
\usepackage{amssymb,amsmath}
\usepackage{pgfplots}
\usepgfplotslibrary{statistics}
\usepackage{booktabs} 
\usepackage{array} 
\usepackage{paralist} 
\usepackage{verbatim}
\usepackage{wrapfig}
\usepackage[labelfont=bf]{caption}
\usepackage{subcaption}

\newcommand{\hlbox}[1]{%
  \smallskip\begin{center}
  \fboxrule1pt\fboxsep3pt\fcolorbox{black!45}{black!8}{%
  \begin{minipage}{.96\linewidth}#1\end{minipage}}
  \end{center}\smallskip}

\allowdisplaybreaks

\newcommand{\MinTime}{\emph{min time}\xspace}

\newcommand\recht\operatorname

\begin{document}
\title{Attack time analysis in dynamic attack trees via integer linear programming\thanks{This research has been partially funded   by ERC Consolidator grant 864075 CAESAR and the European Union’s Horizon 2020 research and innovation programme under the Marie Skłodowska-Curie grant agreement No. 101008233.}}
%
%
\author{Milan Lopuhaä-Zwakenberg\inst{1} \and
Mariëlle Stoelinga\inst{1,2}}
\authorrunning{M. Lopuhaä-Zwakenberg \& M. Stoelinga}
%
\institute{University of Twente, the Netherlands \email{\{m.a.lopuhaa,m.i.a.stoelinga\}@utwente.nl} \and Radboud University, the Netherlands}

\pgfplotstableread{
X Y
0.5000   -0.7455
1.5000   -0.6021
2.5000   -0.3912
3.5000   -0.1160
4.5000    0.2150
5.5000    0.7804
6.5000    1.5589
7.5000    2.0527
}{\mytablenaive}

\pgfplotstableread{
X Y 
0.5000   -0.6922
1.5000   -0.3912
2.5000   -0.1530
3.5000   0.0746
4.5000   0.1938
5.5000   0.3461
6.5000    0.4586
7.5000    0.5969
8.5000    0.6203
9.5000    0.7084
10.5000    0.8045
11.5    0.8599
12.5    0.9005
}{\mytablemnaive}

\pgfplotstableread{
X Y 
0.5000   -0.5757
1.5000   -0.3996
2.5000   -0.2321
3.5000   -0.1160
4.5000    0.0280
5.5000    0.0571
6.5000    0.1761
7.5000    0.2372
8.5000    0.3010
9.5000    0.4033
10.5000    0.3785
11.5    0.4073
12.5    0.4319
}{\mytablemilp}

\pgfplotstableread{
X Y 
0.5000   -0.7270
1.5000   -0.3291
2.5000   -0.0503
3.5000   0.1232
4.5000    0.2694
5.5000    0.3897
6.5000    0.5098
7.5000    0.6004
8.5000    0.6824
9.5000    0.7804
10.5000    0.7948
11.5    0.8764
12.5    0.9416
}{\mytablemmilp}

\pgfplotstableread{
X Y
0.5   -0.8062
1.5   -0.5757
2.5   -0.3912
3.5   0.0819
4.5   0.2639
5.5   0.7793
6.5   1.6355
7.5   1.9816
}{\mytablenaivedag}

\pgfplotstableread{
X Y 
0.5000   -0.3912
1.5000   -0.3219
2.5000   -0.1934
3.5000   -0.0986
4.5000   -0.0209
5.5000   0.0571
6.5000    0.1023
7.5000    0.1481
8.5000    0.2024
9.5000    0.2526
10.5000    0.2802
11.5    0.3274
12.5    0.3611
}{\mytablemilpdag}

\pgfplotstableread{
X Y 
0.5000   -0.7270
1.5000   -0.3748
2.5000   -0.1434
3.5000   0.0389
4.5000    0.2003
5.5000    0.3492
6.5000    0.4839
7.5000    0.6570
8.5000    1.0607
9.5000    1.3670
10.5000    1.9460
11.5    2.1194
12.5    2.4654
}{\mytablemnaivedag}

\pgfplotstableread{
X Y 
0.5000   -0.4737
1.5000   -0.2041
2.5000   -0.0353
3.5000   0.1383
4.5000    0.2430
5.5000    0.3582
6.5000    0.4418
7.5000    0.5362
8.5000    0.6105
9.5000    0.6956
10.5000    0.7608
11.5    0.8283
12.5    0.8928
}{\mytablemmilpdag}

\maketitle              
\begin{abstract}
Attack trees (ATs) are an important tool in security analysis, and an important part of AT analysis is computing metrics. However, metric computation is NP-complete in general. In this paper, we showcase the use of mixed integer linear programming (MILP) as a tool for quantitative analysis. Specifically, we use MILP to solve the open problem of calculating the \emph{min time} metric of dynamic ATs, i.e.,  the minimal time to attack a system. We also present two other tools to further improve our MILP method: First, we show how the computation can be sped up by identifying the modules of an AT, i.e. subtrees connected to the rest of the AT via only one node. Second, we define a general semantics for dynamic ATs that significantly relaxes the restrictions on attack trees compared to earlier work, allowing us to apply our methods to a wide variety of ATs. Experiments on a synthetic testing set of large ATs verify that both the integer linear programming approach and modular analysis considerably decrease the computation time of attack time analysis.

\keywords{Attack trees \and Quantitative analysis \and Optimization \and Mixed integer linear programming.}
\end{abstract}

\section{Introduction} \label{sec:int}

\begin{wrapfigure}[6]{r}{3cm}
\centering
\vspace{-3em}
\includegraphics[width=2.5cm]{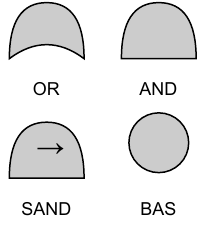}
\end{wrapfigure}
\textbf{(Dynamic) attack trees.} 
Attack trees (ATs) are a prominent methodology in security analysis. They facilitate
security specialists in identifying, documenting, analyzing and prioritizing
 (cyber) risks.
An AT is a hierarchical diagram that describes a system's vulnerabilities to an adversary's attacks.
Despite their name, ATs are rooted directed acyclic graphs. Roots of ATs represent the adversary's goal, while the leaves represent basic attack steps (BAS) undertaken by the adversary. Each internal root is labeled with a gate, determining how its activation depends on that of its children.
Standard ATs (SATs) feature only $\mathtt{OR}$ and $\mathtt{AND}$ gates, 
but many extensions have been introduced to describe more elaborate attack scenarios \cite{Kordy2014}. One of the most prominent extensions are \emph{dynamic ATs} (DATs) \cite{Jhawar2015}. DATs introduce a $\mathtt{SAND}$ (sequential $\mathtt{AND}$) gate, which is activated only when its children are activated sequentially in the correct order. By contrast, an $\mathtt{AND}$-node's children can be activated in parallel. An example is given in Figure \ref{fig:troy}.

\begin{wrapfigure}[24]{r}{4cm}
    \centering
\vspace{-2.5em}
\includegraphics[width=4cm]{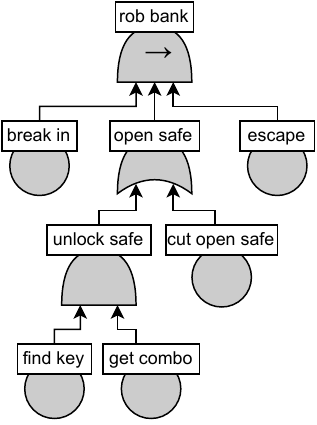}
\caption{A DAT for a bank robbery \cite{Arnold2014}. To rob a bank, attackers must break in, open the safe, and escape (in that order). The safe is opened by cutting it open, or by unlocking via obtaining the key and combination.}
 \label{fig:troy}
\end{wrapfigure}

\noindent \textbf{Quantitative analysis.}
Quantitative analysis aims at computing {\em AT metrics}. Such metrics  formalize how well a system performs in terms of security, and are essential when comparing alternatives or making trade-offs. Many such metrics exist, such as the minimal cost, minimal required skill, or maximal damage of a successful attack. This paper focuses on \emph{min time}: the minimal time the adversary needs to perform a successful attack, given the duration of each BAS. This is important, since
 attack success crucially depends on time: attacks that take too long are not viable. Insight in timing behaviors of attacks  is therefore a key to devising effective countermeasures. For instance, a security operations centre is interested in the time difference between the fastest viable attack and its average response time \cite{agyepong2020challenges}. 
 \emph{Min time} is especially relevant in the context of DATs: 
On many metrics, such as cost/probability/skill,
$\mathtt{SAND}$ and $\mathtt{AND}$ gates behave identically. Thus, to compute those metrics, 
algorithms for SATs immediately generalize to DATs. 
It is in the timing behavior that the  difference between 
$\mathtt{SAND}$ and $\mathtt{AND}$
manifests itself, so that novel computation algorithms are needed.


\begin{wrapfigure}[4]{r}{3cm}
\vspace{-5.5em}
\centering
\includegraphics[width=2.5cm]{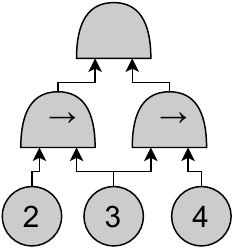}
\end{wrapfigure}
\noindent \textbf{Existing algorithms for \emph{min time}.}
The naive approach to calculting \emph{min time} is to list all attacks that reach the root, and to find the one that takes the least time; clearly this is computationally prohibitive for larger ATs.
%
A tree-shaped DAT can be computed via a bottom-up (BU) algorithm \cite{Mauw2006,Jhawar2015}.
This algorithm works for general attributes (e.g. cost, probability, time), by using appropriate operators at each gate.
For DAG-shaped ATs, the BU algorithm does not always work, because the values in different branches are no longer independent. For SATs this is not a problem because the relevant operators are idempotent \cite{Kordy2018}. In the DAT above, however, the BU algorithm of \cite{Jhawar2015} calculates \emph{min time} as $\max(2+3,3+4) = 7$. However, the only successful attack is the one that activates the three BAS sequentially, and so \emph{min time} equals $2+3+4=9$. Thus to find \MinTime for DAG-shaped DATs new approaches are needed; in \cite{Budde2021}, efficient computation for DAG-shaped DATs is left as an open problem.

\noindent \textbf{Integer linear programming.}
 In this paper, we present a novel method to  calculate \emph{min time} for general DATs based on MILP. We translate calculating \emph{min time} into a real-valued optimization problem, with a set of nonlinear constraints. We rewrite these into linear constraints by introducing auxiliary integer variables at each gate; for SAND-gates this is nontrivial and requires a careful analysis of the semantics, beyond the current literature (see below). Since dedicated solvers exist for MILP, translating attack time analysis into MILP speeds up computation time considerably. 


\noindent \textbf{Modular analysis.}
To improve performance, we combine MILP with modular analysis \cite{Dutuit1996}: we identify \emph{modules} in a DAT, i.e., subDATs whose only connections to the rest of the DAT go via their root.
We prove that \emph{min time} can be computed by analyzing the modules separately; this requires a detailed comparison of the attacks on the larger DAT to the attacks on its modules. If a module is tree-shaped or static, then we can deploy the bottom-up algorithm to further decrease computation time. We integrate these modifications into our MILP algorithm.

\noindent \textbf{Generalized semantics.}
Another point we settle in this paper are generalized semantics for DATs.  
As $\texttt{SAND}$-gates require their children to be executed consecutively, different branches in the DAT may impose conflicting restrictions on the execution orders.
To rule out these conflicts, \cite{Budde2021} imposed well-formedness criteria at the cost of ruling out some satisfiable DATs. Furthermore, the corresponding attack definition was overly restrictive, with some fastest attacks not being recognized. This leads to an overestimation of \emph{min time}. In this work we extend the definition of a (successful) attack so \emph{min time} is correctly defined. This new definition applies to all DATs, not just the well-formed ones.

\noindent \textbf{Experimental validation.} 
For confidentiality reasons industrial DATs are typically not disclosed to the general public \cite{byres2004use,paul2014towards}. Therefore, we create a testing suite of 2400 synthetic DATs,
obtained by combining smaller DATs from the literature via standard DAT composition methods,
and we compare the performance of four methods (modular versus nonmodular and enumerative versus MILP). The experiments show that on larger DATs MILP outperforms enumerative, and modular outperforms nonmodular. The code for the experiments, the generated DATs and the experimental results are available in \cite{code}, and a version with proofs is available at \cite{lopuhaa2021attack}.

\vspace{-0.75em}
\hlbox{\textbf{Contributions.}
Summarized our main contributions are:
\begin{enumerate}
    \item A generalization of the poset semantics of \cite{Budde2021} that 
    significantly relaxes the syntactic constraints on the use of $\mathtt{SAND}$-gates.
    \item A novel algorithm to calculate \emph{min time} for general DATs based on Mixed Integer Linear Programming.
    \item A modularization approach that yields significant speed ups by separately handling fragments of the DAT that are static or tree-shaped.
    \item Extensive experimental validation to evaluate the performance of the algorithms.
\end{enumerate}
}
\vspace{-0.75em}
\section{Related work} \label{sec:rel}

Dynamic ATs were first formally defined in \cite{Jhawar2015}, with series-parallel graphs semantics.
These assume that each node must be activated separately for each of its parents. Effectively, this makes any DAT tree-shaped, which limits the range of scenarios that can be modeled.

Poset-semantics for DATs are used in \cite{Budde2021}; here each node can be activated only once, allowing more scenarios to be modeled. The calculation of time-related metrics such as \emph{min time} on DAG-shaped DATs is left as an open problem.

In \cite{Arnold2015,Kumar2015,kumar2018effective,vitkus2020method,ali2021attack} DATs are modeled as priced-timed automata. This allows for a detailed analysis, including \emph{min time} calculation, by activating nodes from the root either in parallel or sequentially, depending on gate type. However, this approach does not consider satisfiability; hence the \MinTime found via this method can correspond to a non-existing attack. As such, this method only calculates a lower bound to the actual \MinTime. 

Cyber security risks are also analyzed via \emph{time-to-compromise} \cite{mcqueen2006time}. This assigns an (exponential) probability distribution to the failure time of each component, from which one finds the system failure pdf. This approach can be extended to consider different attack scenarios \cite{rencelj2022estimating}. The current paper's DAT approach allows for a more systematic way of studying different attack scenarios, but we do not consider probabilistic data. Another way to incorporate stochastics is to consider \emph{Bayesian fault trees} \cite{franccois2016bayesian,meyur2020bayesian}, in which a node's activation depends probabilistically on that of its children. This allows for more detailed modelling, but analysis is considerably more complicated: instead of a single \MinTime metric, there is a Pareto front of attack time and attack success probability. Incorporating probability in these manners would be interesting for future work. 

Time analysis of DATs falls into the wider framework of quantitative analysis on ATs. Existing approaches either focus on a single metric \cite{Arnold2014,buldas2006rational,byres2004use} or they develop methods that apply to general classes of metrics \cite{Mauw2006,Kordy2018,Budde2021}. The latter case typically use algebraic structures like semirings, defining the metric in terms of operators which are assumed to have certain properties.

\section{Dynamic attack trees} \label{sec:dat}

This section reviews the definition of DATs, and develops their semantics and the \emph{min time} metric. The notation introduced throughout the paper is summarized in Table \ref{tab:nota}. The following definition of a DAT is from \cite{Budde2021}.

\begin{definition}
A \emph{dynamic attack tree} (DAT) is a rooted directed acyclic graph $T = (N,E)$ where each node $v \in N$ has a type $\gamma(v) \in \{\mathtt{BAS},\mathtt{OR},\mathtt{AND},\mathtt{SAND}\}$ such that $\gamma(v) = \mathtt{BAS}$ if and only if $v$ is a leaf, and every node $v$ with $\gamma(v) = \mathtt{SAND}$ has an ordering of its set of children.
\end{definition}

Note that a DAT is not necessarily a tree. If it is, we call it \emph{tree-shaped}. The root is denoted $\recht{R}_T$. For $\gamma \in \{\mathtt{BAS},\mathtt{OR},\mathtt{AND},\mathtt{SAND}\}$, we write $N_\gamma$ for the set of nodes $v$ with $\gamma(v) = \gamma$. The (po)set of children of $v$ is denoted $\recht{ch}(v)$. If $\gamma(v) = \mathtt{SAND}$ and $v$ has (ordered) children $v_1,\ldots,v_n$, we write $v = \mathtt{SAND}(v_1,\ldots,v_n)$ for convenience. We do the same for $\mathtt{OR}$ and $\mathtt{AND}$, where the ordering of the children does not matter. We write $T_v$ for the subDAG consisting of all \emph{descendants} of $v$, i.e. all $v'$ for which there is a path from $v$ to $v'$, including $v$ itself. Furthermore, we let $B_v$ be the set of descendants of $v$ in $N_{\mathtt{BAS}}$. DATs can be represented graphically as in Fig.~\ref{fig:troy}.

A dynamic attack tree codifies the ways an attacker can make a system fail by executing the \emph{basic attack steps}, i.e., the nodes in $N_{\mathtt{BAS}}$. A non-$\mathtt{BAS}$ node is reached depending on its children, where $\mathtt{OR}$ and $\mathtt{AND}$ have the expected meaning, and a $\mathtt{SAND}$-node is reached if all children are reached in their given order. The adversary's goal is to reach $\recht{R}_T$. These semantics are defined in Section \ref{ssec:sem}.

\begin{table}[t]
\centering
\begin{tabular}{cccccc}
Notation & Meaning & Section & Notation & Meaning & Section \\
\hline
$T = (N,E)$ & Dynamic attack tree & \ref{sec:dat} & $\recht{mt}(T,d), \recht{mt}(T)$ & \MinTime of DAT $T$ & \ref{ssec:mt}\\
$\gamma(v)$ & Type of $v$ & \ref{sec:dat} & $\mathcal{F}_T$ & Time assignments of $T$ & \ref{sec:milp}\\
$\recht{R}_T$ & Root of $T$ & \ref{sec:dat} & $M$ & \MinTime upper bound & \ref{sec:milp}\\
$N_{\mathtt{BAS}}$ & $\{v \in N \mid \gamma(v) = \mathtt{BAS}\}$ & \ref{sec:dat} & $Z^v_i$ & Consecutive BAS pairs & \ref{sec:milp}\\
$T_v$ & subDAT with root $v$ & \ref{sec:dat}  & $x^v_i, y^v,z^v_{i,a,a'}$ & Auxiliary MILP variables & \ref{sec:milp}\\
$B_v$ & Set of BAS of $T_v$ & \ref{sec:dat} & $n_v$ & Number of children of $v$ & \ref{sec:milp}\\
$(\mathcal{A}_T,\leq)$ & Poset of attacks on $T$ & \ref{ssec:sem} & $T_v$ & Sub-DAT with root $v$ & \ref{sec:mod}\\
$\mathcal{S}_T$ & Successful attacks & \ref{ssec:sem} & $\tilde{v}$ & BAS replacement for $T_v$ & \ref{sec:mod}\\
$\recht{t}(\mathcal{O},d), \recht{t}(\mathcal{O})$ & Time of attack $\mathcal{O}$ & \ref{ssec:mt} &
$T^v$ & $T$ with $T_v$ replaced by $\tilde{v}$ & \ref{sec:mod}
\end{tabular}
\caption{Notation used in this paper.} \label{tab:nota}
\vspace{-2em}
\end{table}

In the literature, two interpretations of nodes with multiple parent nodes exist, affecting both semantics and metrics. In the first interpretation, \emph{multiple activation} (MA), \cite{Mauw2006,Jhawar2015,Widel2019} each BAS can be activated multiple times, and every parent of a node requires its own activation of that node. Thus $\mathtt{SAND}(a,a)$ succeeds only if $a$ is activated twice consecutively. By adding a copy of each node for each of its parents, any DAT can be transformed into a tree-shaped one with equivalent semantics and metrics. As a result, metrics can be calculated quickly via a bottom-up algorithm \cite{Fila2020}, but MA cannot adequately model systems in which one action has multiple independent consequences.

In \emph{single activation} (SA) \cite{Budde2021,Jurgenson2008} each BAS is executed at most once, and a node only needs to be activated once to count as an input for all its parents. In SA $\mathtt{SAND}(a,a)$ cannot be satisfied, because $a$ cannot be activated before itself. SA is able to describe a much wider range of systems; although every SA representation can be turned into an equivalent MA representation, this process is both computationally expensive as it is done by writing the corresponding boolean function in disjunctive normal form. This rewriting also loses the meaning of the intermediate nodes in the DAT, which typically represent intermediate attacker goals.
We therefore choose to analyze DATs under the SA interpretation; since every DAT is equivalent to a tree-shaped one under MA and MA and SA coincide on trees, SA can model every scenario that MA can.

\subsection{Semantics} \label{ssec:sem}

We discuss DAT semantics, extending \cite{Budde2021}. An attack consists of a set $A$ of attacker-activated BAS, and a strict partial order $\prec$, where $a \prec a'$ means $a$ is executed before $a'$.

\begin{definition} \label{def:attack}
The set $\mathcal{A}_T$ of \emph{attacks} on $T$ is the set of strictly partially ordered sets  $\mathcal{O} = (A,\prec)$, where $A \subseteq N_{\mathtt{BAS}}$. This set has a partial order $\leq$ given by $\mathcal{O} \leq \mathcal{O}'$, for $\mathcal{O} = (A,\prec)$ and $\mathcal{O}' = (A',\prec')$, if and only if $A \subseteq A'$ and ${\prec} \subseteq {\prec'}$.
\end{definition}

We are interested in successful attacks, i.e., attacks that manage to reach the root. Successful attacks, and the semantics of $T$, are defined as follows:

\begin{definition} \label{def:success}
Let $v$ be a node. We say that an attack  $\mathcal{O} = (A,\prec)$ \emph{reaches} $v$ if:
\begin{enumerate}
\item $v \in N_{\mathtt{BAS}}$ and $v \in A$;
\item $v = \mathtt{OR}(v_1,\ldots,v_n)$ and $\mathcal{O}$ reaches at least one of the $v_i$;
\item $v = \mathtt{AND}(v_1,\ldots,v_n)$ and $\mathcal{O}$ reaches all of the $v_i$;
\item $v = \mathtt{SAND}(v_1,\ldots,v_n)$ and  $\mathcal{O}$ reaches all of the $v_i$, 
and for all $a \in A \cap B_{v_i}$, $a' \in A \cap B_{v_{i+1}}$ one has $a \prec a'$.
\end{enumerate}
$A$ is \emph{successful} if it reaches $\recht{R}_T$. The \emph{semantics} of $T$ is the set $\mathcal{S}_T$ of successful attacks on $T$.
\end{definition}

\begin{wrapfigure}[6]{r}{3cm}
\centering
\vspace{-4em}
\includegraphics[width=2.5cm]{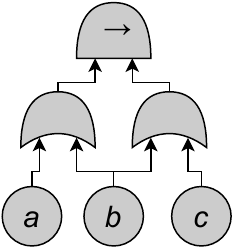}
\end{wrapfigure}
A $\mathtt{SAND}$-gate $v = \mathtt{SAND}(v_1,\ldots,v_n)$ is only reached if all of the BAS of $v_i$ have been (sucesfully) executed before any of the BAS of $v_{i+1}$ has started. By contrast, an $\mathtt{AND}$-gate allows its children to be executed in parallel. Contrary to the static case (without $\mathtt{SAND}$-gates), it is possible that $\mathcal{S}_T = \varnothing$. For example, $\mathcal{S}_{\mathtt{SAND}(a,a)}=\varnothing$. Also, being successful is \emph{not} monotonous on the set of attacks, i.e., it is possible that $\mathcal{O}$ is successful while $\mathcal{O}'$ is not, even if $\mathcal{O} \leq \mathcal{O}'$. For instance, in the DAT above $(\{a,c\},\{(a,c)\})$ is a successful attack, but $(\{a,b,c\},\{(a,c)\})$ is not. Note that unlike the situation for SATs, a gate's activation does not simply depend on the activation of its children, but also on the relative order on the BAS associated to these children; this encodes the iming information essential to DATs.

\begin{wrapfigure}[8]{r}{2cm}
\centering
\vspace{-2em}
\includegraphics[width=1.5cm]{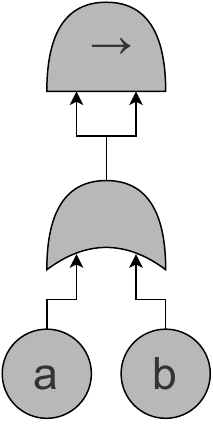}
\end{wrapfigure}
Definition \ref{def:success} is not the only way one might define the semantics of DATs. In fact, our semantics are based on those of \cite{Budde2021}, but differ on certain DATs; see Section \ref{ssec:carlos}. We have chosen to interpret the $\mathtt{SAND}$-gate in a strict matter, so that it is activated only if the entirety of the attack on $v_i$ has finished before the attack on $v_{i+1}$ is started; in particular, $v_i$ and $v_{i+1}$ cannot share activated BAS, which may be considered unwanted behaviour. There are also other approaches, which unfortunately have other problems. For instance, one could define succesful attacks bottom-up in a compositional fashion, defining $\mathcal{O}$ to reach $\mathtt{SAND}(v_1,v_2)$ if there exists attacks $\mathcal{O}_1,\mathcal{O}_2$ such that $\mathcal{O}$ is the parallel composition of $\mathcal{O}_1$ and $\mathcal{O}_2$. However, under such a definition the AT above $(\{a,b\},\{(a,b)\})$ is a succesful attack, whereas in our opinion this AT should not be considered satisfiable. Yet another approach would be to assign a starting and finishing time to each node, similar to what we do in Definition \ref{def:ta}, but this has the disadvantage of being more convoluted as an attack is now a function $N \rightarrow \mathbb{R}$. 

\subsection{The \emph{min time} metric} \label{ssec:mt}

\emph{Min time} is the minimal time it takes to perform a successful attack on a given DAT. While other metrics exist for DATs, \emph{min time} is a fundamental time metric, and calculating it efficiently for non-tree-shaped DATs is an open problem \cite{Budde2021}.

\emph{Min time} is defined as follows: There is a \emph{duration function} $d\colon N_{\mathtt{BAS}} \rightarrow \mathbb{R}_{\geq 0}$, with $d(a)$ denoting the time it takes to execute $a$. If $a \prec a'$, then the BAS $a'$ can only be started once $a$ has been completed, while $a$ and $a'$
can be activated in parallel if such a relation does not exist. As such, we can define the total duration of an attack $\recht{t}(\mathcal{O},d)$ and \emph{min time} $\recht{mt}(T,d)$ as
\begin{equation*}
\recht{t}(\mathcal{O},d) = \max_{\substack{C\textrm{ max. chain} \\ \textrm{in }\mathcal{O}}} \sum_{a \in C} d(a), \quad \quad \recht{mt}(T,d) = \min_{\mathcal{O} \in \mathcal{S}_T} \recht{t}(\mathcal{O},d)
\end{equation*}
where the maximum is taken over the maximal chains (i.e., maximal linearly ordered subsets) of the strict poset $\mathcal{O}$. We will often omit $d$ from the notation and write $\recht{t}(\mathcal{O})$ if there is no confusion. Note that $\recht{t}$ is monotonous: if $\mathcal{O} \leq \mathcal{O}'$
one has $\recht{t}(\mathcal{O}) \leq \recht{t}(\mathcal{O}')$. Furthermore, $\recht{mt}(T) = \infty$ if $\mathcal{S}_T = \varnothing$. 

\begin{example} \label{ex:troy}
Figure \ref{fig:troy2} depicts the bank robbery DAT of Figure \ref{fig:troy} augmented with durations for the BAS (we take the expected durations from the distributions given in \cite{Arnold2014}). To calculate $\recht{mt}(T)$ one would first need to find $\mathcal{S}_T$. While this set is quite large, because of the monotonicity of $\recht{t}$, the minimum is attained at one of the minimal elements of the poset $(\mathcal{S}_T,\leq)$. There are two minimal attacks, depending on whether the attackers choose to cut open the safe, or unlock it. Abbreviating BAS names, we can represent these minimal attacks as sets of chains as $\mathcal{O}_1 = \{bi \prec cos \prec e\}$ and $\mathcal{O}_2 = \{bi 
\prec fk \prec e, bi \prec gc \prec e\}.$
These have duration $\recht{t}(\mathcal{O}_1) = 1.00+0.67+0.20 = 1.87$ and $\recht{t}(\mathcal{O}_2) = \max(1.00+0.50+0.20,1.00+1.00+0.20) = \max(1.70,2.20) = 2.20$. It follows that $\recht{mt}(T) = \min(1.87,2.20) = 1.87$.
\end{example}

\begin{wrapfigure}[13]{r}{4.5cm}
\vspace{-4em}
    \centering
\includegraphics[width=4cm]{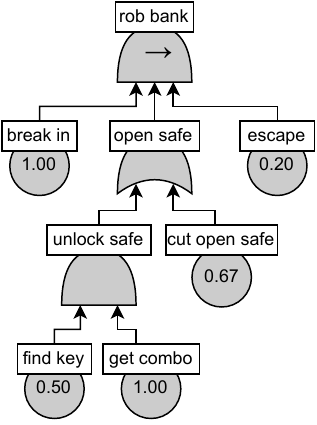}
\caption{The bank robbery DAT of Figure \ref{fig:troy} augmented with durations.}
    \label{fig:troy2}
\end{wrapfigure}
In the multiple activation scenario, \emph{min time} can be calculated by reshaping a DAT into its canonical form \cite{Jhawar2015}, from which \emph{min time} is easily calculated. However, this technique does not carry over to our formalism, as in the single activation scenario a canonical form does not exist.

\subsection{Relation to semantics of \cite{Budde2021}} \label{ssec:carlos}

In \cite{Budde2021} attacks are called attacks only if they satisfy the ordering constraints imposed by \emph{all} $\mathtt{SAND}$-gates. This is defined only for \emph{well-formed} DATs, i.e., all these constraints are simultaneously satisfiable. More formally, that work only considers attacks that we call \emph{full} in the following definition.

\begin{definition}
Let $T$ be a DAT. Define a relation $\sqsubseteq'$ on $N_{\mathtt{BAS}}$ by $a \sqsubseteq' a'$ iff there exists a node $v = \mathtt{SAND}(v_1,\ldots,v_n)$ and an $i < n$ such that $a \in B_{v_i}$ and $a' \in B_{v_{i+1}}$. Let $\sqsubseteq$ be the transitive closure of $\sqsubseteq'$. Then $T$ is \emph{well-formed} if $\sqsubseteq$ is a strict partial order. An attack $(A,\prec)$ on a well-formed DAT is \emph{full} if ${\prec} = {\sqsubseteq}|_A$, the restriction of $\sqsubseteq$ to $A$.
\end{definition}

\begin{wrapfigure}[5]{r}{2.5cm}
\centering
\vspace{-4.5em}
\includegraphics[width=2cm]{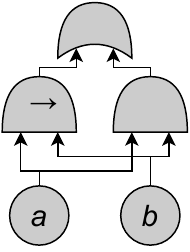}
\end{wrapfigure}
However, not all attacks will be full, because an attack may not need to reach all $\mathtt{SAND}$-nodes in order to reach the root, and non-reached nodes should not put restrictions on attacks. Consider the well-formed DAT on the right. Only $(\{a,b\},\{(a,b)\})$ is a full successful attack. However, $(\{a,b\},\varnothing)$ is a successful attack as well. Hence non-full attacks are needed to fully describe the semantics of well-formed DATs, which motivates Def.~\ref{def:attack}. Furthermore, our definition defines the semantics of general DATs, not just the well-formed ones. 

\section{An MILP approach to \emph{min time}} \label{sec:milp}



This section describes a novel method to compute $\recht{mt}(T)$ based on mixed-integer linear programming (MILP). Although MILP is NP-complete, a number of good heuristics and solvers exist specifically for MILP, which can result in a low computation time. We first show that \emph{min time} can be found by solving an optimization problem in Theorem \ref{thm:ta}, and then we describe how that optimization problem can be rewritten into the MILP framework.


The building block of the new approach is the notion of {\em time assignment}, which assigns to each node a completion time $f_v$ that respects all timing constraints in the DAT. If $f_v = \infty$ then $v$ is not reached at all. The formal definition is stated below; recall that $B_v$ is the set of BAS-descendants of $v$, and $N$ the set of nodes in the attack tree. 

\begin{definition} \label{def:ta}
Let $T$ be a DAT. For a node $v$ with children $v_1,\ldots,v_n$ and $i < n$, define $Z^v_i := B_{v_i} \times B_{v_{i+1}}$. A \emph{time assignment} is a vector $f \in [0,\infty]^N$ satisfying:
\begin{enumerate}
    \item For each $a \in N_{\mathtt{BAS}}$ one has $f_a\geq d(a)$;
    \item For each $v = \mathtt{OR}(v_1,\ldots,v_n)$ one has $f_v \geq \min_i f_{v_i}$;
    \item For each $v = \mathtt{AND}(v_1,\ldots,v_n)$ one has $f_v \geq \max_i f_{v_i}$;
    \item For each $v = \mathtt{SAND}(v_1,\ldots,v_n)$, the following must hold:
    \begin{enumerate}
        \item it holds that $f_v \geq f_{v_n}$;
        \item If there is a $i \leq n$ such that $f_{v_i} = \infty$, then $f_v = \infty$;
        \item If there exist $i <n$ and $(a,a') \in Z_i^v$ such that $f_{a'}-d(a') < f_a < \infty$, then $f_v = \infty$.
    \end{enumerate}
\end{enumerate}
The set of all time assignments for $T$ is denoted $\mathcal{F}_T$.

\end{definition}


The $\mathtt{SAND}$-conditions can be understood as follows. 4a) tells us that $v$ cannot be reached before $v_n$, and 4b) tells us that $v$ cannot be reached if any of its children is not reached. 4c) conveys that whenever there is an $a \in B_{v_i}$ that is activated (i.e., $f_a < \infty$), then in order for $v$ to be activated, one must have $f_{a'} - d(a') \geq f_a$ for all $a' \in B_{v_{i+1}}$. Since $f_{a'} - d(a')$ is the starting time of $a'$, this means that $a'$ must be started after $a$ is finished activating. It is more subtle than simply requiring $f_{a'} - d(a') \geq f_a$ for all $(a,a') \in Z_i^v$; that would ensure that all $\mathtt{SAND}$-gates impose ordering restrictions, not just those that are activated.

Note that $f_{a'} - d(a')$ is the starting time of a BAS $a'$, so 4c) tells us that $v$ is only reached if the BAS-descendants
of $v_{i+1}$ are started once those of $v_i$ have been completed. We allow for a delay in completing node $v$, even when enough of its children have been completed. Time assignments relate to \emph{min time}:

\begin{theorem} \label{thm:ta} $\recht{mt}(T) = \min_{f \in \mathcal{F}_T} f_{\recht{R}_T}$.
\end{theorem}

This result allows us to calculate $\recht{mt}(T)$ by solving the following optimization problem.
\begin{align}
\recht{minimize}_{f \in [0,\infty]^N} \ \ & f_{\recht{R}_T} \quad \quad \quad \textrm{s.t.} \ \ f \in \mathcal{F}_T. \label{eq:opt}
\end{align}
This is not a linear problem, due to the nonlinear constraints of Definition \ref{def:ta}. We use auxiliary integer variables to linearize these constraints.
First, we need to get rid of the $\infty$ in Definition \ref{def:ta}, which we do by replacing it with a suitably large real number. Define the constant $M = 1+\sum_{a \in N_{\mathtt{BAS}}} d(a).$
The following lemma shows that if $T$ is satisfiable, then to minimize \eqref{eq:opt} one can focus on the $f$ with $f_v \in [0,M-1] \cup \infty$.

\begin{lemma} \label{lem:maxm}
There is an $f$ minimizing \eqref{eq:opt} for which $\forall v\colon f_v \in [0,M-1] \cup \infty$.
\end{lemma}

This shows that we can use $M$ to play the role of $\infty$ where necessary. We enforce this by demanding $f_v \in [0,M]$, and we interpret $f_v = M$ to mean that $v$ is not reached. For a node $v$, let $n_v$ be its number of children, which are denoted $v_1,\ldots,v_{n_v}$. We then use standard MILP techniques \cite{Chen2011} to rewrite Definition \ref{def:ta}.

To rewrite the $\mathtt{OR}$-condition, we introduce an auxiliary binary variable $x_i^v$ for each
$v \in N_{\mathtt{OR}}$ and each $i \leq n_v$. The purpose of $x_i^v$ is to represent the truthfulness of the statement ``$i = \arg\min_{i'} f_{v_{i'}}$''. 
We can then represent $f_v \geq \min_i f_{v_i}$ by
\begin{align*}
\sum_{i\leq n_v} x_i^v &\geq 1, & \forall i \leq n_v\colon f_v &\geq f_{v_i} + M(x_i^v-1).
\end{align*}
The latter is automatically satisfied if $x_i^v = 0$, and reduces to $f_v \geq f_{v_i}$ if $x_i^v = 1$. The former ensures that the latter must happen for at least one $i$, so together these encode $f_v \geq \min_i f_{v_i}$. The condition for $\mathtt{AND}$-gates can be rewritten as $\forall i \leq n_v\colon f_v \geq f_{v_i}$.

Finally, we consider $\mathtt{SAND}$-gates. For $v \in N_{\mathtt{SAND}}$, we introduce an auxiliary binary variable $y^v$ that encodes ``$\exists i < n\colon f_{v_i} = \infty$ or $\exists i\exists (a,a') \in Z_i^v\colon f_{a'}-d(a') < f_a < \infty$.'' Then we can write Definition \ref{def:ta}.4 as
$f_{v} \geq f_{v_{n_v}}, 
f_v \geq My^v.$ To ensure $y^v = 1$ whenever one of the $f_{v_i}$ equals $\infty$, we add the constraint $\forall i < n_v\colon y^v \geq \tfrac{1+f_{v_i}-M}{M}$,
which forces $y^v = 1$ only when $f_{v_i} > M-1$. Furthermore, to ensure $y^v = 1$ whenever some $a,a'$ satisfy $f_{a'}-d(a') < f_a$, we would like to add the constraint
\begin{align} \label{eq:sand1}
&\forall i < n_v \forall (a,a') \in Z_i^v\colon y^v \geq \min\left\{\tfrac{f_a-f_{a'}+d(a')}{M},\tfrac{M-f_a}{M}\right\}. 
\end{align}
This forces $y^v = 1$ only when both $f_{a'}-d(a') < f_a$ and $f_a < M$. To get rid of the minimum, we introduce an auxiliary variable $z^v_{i,a,a'}$ for each $i < n_v$ and $(a,a') \in Z_i^v$ as we did for the $\mathtt{OR}$-condition. We then replace \eqref{eq:sand1} with
\begin{align*}
\forall i < n_v \forall (a,a') \in Z_i^v\colon y^v &\geq \tfrac{f_a-f_{a'}+d(a')}{M}-z^v_{i,a,a'},\ \ y^v \geq\tfrac{M-f_a}{M}-(1-z^v_{i,a,a'}).
\end{align*}

Taking all of this together, it can be shown that the constraint $f_v \in [0,M]$ holds automatically for all `reasonable' $f$ (i.e., if this does not hold for $f$, then $f$ will not minimize $f_{\recht{R}_T}$) and can be replaced by $f_v \in \mathbb{R}$. We then find that the optimization problem \eqref{eq:opt} can be rewritten into the following MILP problem of Figure \ref{fig:milp}. Note that this optimization returns an $f$ with $f_{\recht{R}_T} \leq M-1$ if and only if $\mathcal{S}_T \neq \varnothing$. Hence this optimization can also be used to determine whether $T$ can successfuly be attacked.

We note that this is not the only way to encode \MinTime analysis into a MILP problem; for instance, instead of using the constant $M$, one could introduce an additional binary variable per node that denotes whether the node is activated or not. We chose for this approach since this ensures we need fewer optimization variables, even though this means that some equations such as \eqref{eq:sand1} are less intuitive. Note that we get quadratically many constraints above, which is a consequence of the fact that we get a constraint for every pair $(a,a')$ in Definition \ref{def:success}.4.


\begin{figure*}
\begin{align*}
&\recht{minimize } f_{\recht{R}_T} \recht{\ subject \ to:}\\
&\forall v \in N&\colon &f_v \in \mathbb{R},\\
&\forall a \in N_{\mathtt{BAS}}&\colon &f_a \geq d(a),\\
&\forall v \in N_{\mathtt{OR}},\ \forall i \leq n_v&\colon &x^v_i \in \{0,1\},\\
&\forall v \in N_{\mathtt{OR}},\ \forall i \leq n_v&\colon &f_v \geq f_{v_i}+M(x^v_i-1),\\
&\forall v \in N_{\mathtt{OR}}&\colon &\sum_{i\leq n_v} x_i^v \geq 1,\\
&\forall v \in N_{\mathtt{AND}},\ \forall i \leq n_v&\colon &f_v \geq f_{v_i},\\
&\forall v \in N_{\mathtt{SAND}}&\colon &y^v \in \{0,1\},\\
&\forall v \in N_{\mathtt{SAND}},\ \forall i < n_v,\  \forall (a,a') \in Z^v_i&\colon &z^v_{i,a,a'} \in \{0,1\},\\
&\forall v \in N_{\mathtt{SAND}}&\colon&f_v \geq f_{v_{n_v}},\\ 
&\forall v \in N_{\mathtt{SAND}}&\colon &f_v \geq My^v,\\
&\forall v \in N_{\mathtt{SAND}},\ \forall i < n_v&\colon &y^v \geq \tfrac{1+f_{v_i}-M}{M},\\
&\forall v \in N_{\mathtt{SAND}},\ \forall i < n_v,\ \forall (a,a') \in Z^v_i&\colon &y^v \geq \tfrac{f_a-f_{a'}+d(a')}{M}-z^v_{i,a,a'} ,\\
&\forall v \in N_{\mathtt{SAND}},\ \forall i < n_v,\ \forall (a,a') \in Z^v_i&\colon &y^v \geq \tfrac{M-f_a+1}{M}+z^v_{i,a,a'}-1.\\
\end{align*}
\vspace{-2.5em}
\caption{The MILP problem for calculating \emph{min time}.} \label{fig:milp}
\end{figure*}





\section{Computation time reduction} \label{sec:mod}

In this section, we introduce an algorithm reducing the complexity of computing $\recht{mt}(T)$. The algorithm consists of two components: First, we show that a bottom-up algorithm from \cite{Jhawar2015} can be used to calculate \emph{min time} for static (no $\mathtt{SAND}$-gates) and tree-shaped DATs. As the state of the art method, based on binary decision diagrams \cite{Budde2021}, has exponential complexity, and the bottom-up algorithm has linear complexity, this is a big improvement. Second, we split up the calculation of \emph{min time} into parts by identifying the \emph{modules} of a DAT, i.e. subDAGs that are connected to the rest of the DAT via only one node. 


\subsection{Bottom-up computation}

\begin{wrapfigure}[11]{r}{6cm}
\vspace{-6em}
	\KwIn{Dynamic attack tree $T$, duration vector $d \in \mathbb{R}^{N_{\mathtt{BAS}}}$}
	\KwOut{Potential min time $\recht{mt}(T,d)$.}
	\BlankLine
	\uIf{$\gamma(v) = \mathtt{BAS}$}{%
		\Return{$d(v)$}
	} \uElseIf{$\gamma(v) = \mathtt{OR}$}{%
		\Return{$\min_{v' \in \recht{ch}(v)} \mathtt{MT\text{-}{BU}}(T_{v'},d|_{B_{v'}})$}
	} \uElseIf{$\gamma(v) = \mathtt{AND}$}{%
		\Return{$\max_{v' \in \recht{ch}(v)} \mathtt{MT\text{-}{BU}}(T_{v'},d|_{B_{v'}})$}
	} \Else(\tcp*[h]{$\gamma(v) = \mathtt{SAND}$}) {%
		\Return{$\sum_{v' \in \recht{ch}(v)} \mathtt{MT\text{-}{BU}}(T_{v'},d|_{B_{v'}})$}
	}
	\caption{$\mathtt{MT\text{-}{BU}}$ for a DAT $T$.}
	\label{alg:BU}
\end{wrapfigure}

An important tool is the algorithm $\mathtt{MT\text{-}BU}$ introduced in \cite{Jhawar2015} presented in Algorithm \ref{alg:BU}. It attempts to calculate $\recht{mt}(T)$ by traversing $T$ bottom-up, which only has linear time complexity and is significantly faster than the MILP approach of Figure \ref{fig:milp}. For tree-shaped $T$ it calculates \MinTime correctly, but for DAGs it fails to account for the fact that two children of a node may share BAS, which may be counted double. However, this double counting is only an issue for $\mathtt{SAND}$-gates, as the operators $\min$/$\max$ of $\mathtt{OR}$/$\mathtt{AND}$-gates are idempotent, i.e., $\min(x,x) = \max(x,x) = x$. This was first realized in \cite{Kordy2018}, for attack-defense trees under different semantics. However, \MinTime based on these \emph{set semantics} can be proven to be equivalent to our definition in \ref{ssec:mt}, yielding the following result:

\begin{theorem} \emph{(\cite{Jhawar2015,Kordy2018})}
If $T$ is tree-shaped or static, then $\mathtt{MT\text{-}BU}$ calculates $\recht{mt}(T)$.
\end{theorem}

\subsection{Modular analysis}

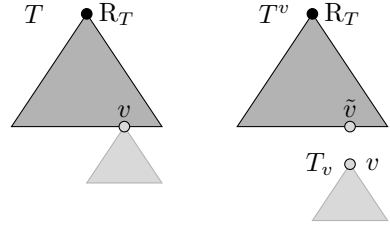
\begin{wrapfigure}[8]{r}{6cm}
\centering
\vspace{-4.5em}
\begin{tikzpicture}
\draw[fill = black!30] (0,0) -- (-1,-1.5) -- (1,-1.5) -- cycle;
\draw[black!30, fill = black!15] (0.5,-1.5) -- (1,-2.25) -- (0,-2.25) -- cycle;
\filldraw (0,0) circle (2pt);
\draw[fill = black!15] (0.5,-1.5) circle (2pt);
\draw (-0.7,0) node {$T$};
\draw (0.4,0) node {$\recht{R}_T$};
\draw (0.5,-1.3) node {$v$};
\draw[fill = black!30] (3,0) -- (2,-1.5) -- (4,-1.5) -- cycle;
\filldraw (3,0) circle (2pt);
\draw[fill = black!15] (3.5,-1.5) circle (2pt);
\draw (2.5,0) node {$T^v$};
\draw (3.4,0) node {$\recht{R}_T$};
\draw (3.5,-1.25) node {$\tilde{v}$};
\draw[black!30, fill = black!15] (3.5,-2) -- (4,-2.75) -- (3,-2.75) -- cycle;
\draw[fill = black!15] (3.5,-2) circle (2pt);
\draw (3.8,-2) node {$v$};
\draw (3.1,-2) node {$T_v$};
\end{tikzpicture}
\caption{Modular analysis.} \label{fig:mod}
\end{wrapfigure}
Algorithm \ref{alg:BU} only reduces complexity in the two relatively rare cases where the DAT is static or tree-shaped. However, it is possible to also reduce complexity when $T$ is only partially static and/or tree-shaped. A well-established method in studying DATs is to consider the \emph{modules} of $T$:

\begin{definition}\emph{\cite{Dutuit1996}}
A \emph{module} is a node $v \in N \setminus N_{\mathtt{BAS}}$ such that all paths from $T\setminus T_v$ to $T_v$ pass through $v$.
\end{definition} 


The root of $T$ is always a module. If $v$ is a module, then $v$ is the only node within $T_v$ with parents outside of $T_v$.
Hence we can create a tree $T^v$ by replacing $T_v$ within $T$ by a new single BAS $\tilde{v}$; the parents of $\tilde{v}$ in $T^v$ are the parents of $v$ in $T$. Theorem \ref{thm:prune} shows that \emph{min time} can be calculated for $T$ by first calculating it for $T_v$, and then for $T^v$. This is depicted in Figure \ref{fig:mod}.

\begin{theorem} \label{thm:prune}
Let $T$ be a DAT, and let $v$ be a module of $T$. Let $T^v$ be the node obtained by removing $v$ and replacing $v$ itself with a new BAS $\tilde{v}$. Then $\recht{mt}(T,d) = \recht{mt}(T^v,d^v)$where $d^v$ is a duration function for $T^v$ given by
\begin{equation*}
d^v(a) = \begin{cases}
d(a), & \textrm{if } a \in N_{\mathtt{BAS}}\setminus B_v,\\
\recht{mt}(T_v,d|_{B_v}), & \textrm{if } a = \tilde{v}.
\end{cases}
\end{equation*}

\end{theorem}

\begin{wrapfigure}[19]{r}{6cm}
\vspace{-3em}
	\KwIn{Dynamic  AT $T$, duration vector $d \in \mathbb{R}^{N_{\mathtt{BAS}}}$, Algorithm $\mathcal{A}$ to calculate \emph{min time}}
	\KwOut{Min time $\recht{mt}(T)$.}
	\BlankLine
	$\mathcal{V} \leftarrow \mathtt{Module}(T)$;
	
	\While{$\mathcal{V} \neq \varnothing$}{
	Pick $v \in \mathcal{V}$ of minimal height;\;
 
	\uIf{$T_v$ is static}
	{$d^v(\tilde{v}) \leftarrow \mathtt{MT\text{-}{BU}}(T_v,d|_{B_v})$;}
	\Else{$d^v(\tilde{v}) \leftarrow \mathcal{A}(T_v,d|_{B_v})$;}
	\For{$a \in N_{\mathrm{BAS}} \setminus B_v$}{
	$d^v(a) \leftarrow d(a)$;}
	$(T,d) \leftarrow (T^v,d^v)$;
	
	$\mathcal{V} \leftarrow \mathcal{V} \setminus \{v\}$;
	}
	\Return{$d({\recht{R}_T})$} \tcp*[h]{$\recht{R}_T$ is a BAS now}
\caption{$\mathcal{A}_{\mathtt{Mod}}$ for a DAT $T$. The notation $T^v,d^v$ is from Theorem \ref{thm:prune}.}
\label{alg:dattime}
\end{wrapfigure}
While the statement seems intuitively true, the proof requires quite a bit of work as one needs to develop machinery to relate attacks on $T$ (and their minimal chains) to attacks on $T_v$ and $T^v$. Theorem \ref{thm:prune} reduces complexity in two ways: We split the tree into two parts whose total size is the same as the original tree. Since MILP is NP-hard, this can impact computation time. Furthermore, the smaller DAT $T_v$ can be static or tree-shaped, in which case we can use $\mathtt{MT\text{-}{BU}}$.

The resulting algorithm is displayed in Algorithm \ref{alg:dattime}. Here $\mathtt{Module}$ refers to an algorithm that finds the modules of $T$; this can be done with linear time complexity \cite{Dutuit1996}. Algorithm $\mathcal{A}_{\mathtt{Mod}}$ makes use of an algorithm $\mathcal{A}$ that calculates \emph{min time}. For this, one can use naive enumeration or the MILP approach of Figure \ref{fig:milp}, or potentially any new algorithm. Since the calculation of a module's \textit{min time} value depends on its own modules, we act on the lower modules first, so Algorithm \ref{alg:dattime} handles the modules by ascending height. Note that when $T$ is tree-shaped, every inner node is a module, so $\mathcal{A}_{\mathtt{Mod}}$ is equivalent to $\mathtt{MT\text{-}{BU}}$ for any $\mathcal{A}$.

We note that other definitions of \emph{min time}, such as the automata-approach of \cite{Kumar2015} and the multiple-activation definition of \cite{Jhawar2015}, also allow for modular decomposition. However, as these definitions are not compatible with ours, we cannot directly use these results, and we require a novel proof for Theorem \ref{thm:prune}.

\section{Experiments} \label{sec:exp}


This section evaluates the performance of our methods. We compare the MILP approach of Figure \ref{fig:milp} ($\mathtt{MT\text{-}MILP}$) to the enumerative approach ($\mathtt{MT\text{-}Enum}$). For the latter, rather than exhaustively generating all succesful attacks, we generate bottom-up a set of candidate attacks that include all minimal succesful attacks, hence certainly the optimal attack by the monotonicity of $\recht{t}$. For this we generalize the set semantics of \cite{Kordy2018} to dynamic ATs. We also compare $\mathtt{MT\text{-}MILP}$ and $\mathtt{MT\text{-}Enum}$ to their modular counterparts.

\begin{wrapfigure}[15]{r}{5cm}
\centering
\vspace{-3em}
\begin{tabular}{cccc}
Source & $|N|$ & tree & Durations\\
\hline
     \cite{Kumar2015} Fig. 1 & 12 & no & unknown \\
     \cite{Kumar2015} Fig. 8 & 20 & no & unknown\\
     \cite{Kumar2015} Fig. 9 & 12 & no & unknown \\
     \cite{Arnold2015} Fig. 1 & 16 & yes & unknown \\
     \cite{Arnold2014} Fig. 3 & 8 & yes & known \\
     \cite{Arnold2014} Fig. 5 & 21 & yes & known\\
     \cite{Arnold2014} Fig. 7 & 25 & yes & known \\
     \cite{Fraile2016} Fig. 2 & 20 & yes & unknown \\
     \cite{Kordy2018} Fig. 1 & 15 & yes & unknown
\end{tabular}
\caption{DATs from the literature used as building blocks. Trees from \cite{Fraile2016,Kordy2018} are adapted from attack-defense trees.} \label{tab:DATs}
\end{wrapfigure}
Existing methods in the literature are based on series-parallel graphs \cite{Jhawar2015} and priced timed automata \cite{Kumar2015}. Their definitions of \MinTime are not equivalent to ours. In our view, methods with different definitions of \MinTime can only be compared with respect to computation time if one of them is designed to be an approximation or bound of the other; then one can compare the gain in computation time versus the loss in accuracy. However, this is not the case here: the multiple activation definition is fundamentally different, and a DAT constructed under this model represents a system different from the same DAT in the single activation model. Therefore, we cannot directly compare performance to that of existing approaches.

In practice, attack trees can be very large \cite{vigo2014automated,paul2014towards}; however, for confidentiality reasons these are typically not disclosed to the general public \cite{byres2004use,paul2014towards}. Hence to our knowledge no established benchmark suites of DATs exist, and the existing literature typically considers test cases with only $\leq 25$ nodes \cite{Arnold2014,Kumar2015}. For such small DATs, the computation of \MinTime takes less than a second no matter which algorithm is being used, which makes them unsuitable for testing difference in algorithm performance. To address the deficiency of a benchmark suite of large DATs, we create a synthetic set of testing DATs. These are created by combining DATs from the literature into larger ones.
Then, we compare (1) the MILP method $\mathtt{MT\text{-}MILP}$ to the enumerative algorithm $\mathtt{MT\text{-}Enum}$ and (2) the effect of modular analysis on performance time. 

All experiments are performed on a PC with an Intel  Core i7-10750HQ 2.8GHz processor and 16GB memory. All algorithms are implemented in Matlab, and for MILP we use the YALMIP environment \cite{yalmip} to translate the optimization problem into the Gurobi solver \cite{gurobi}, a state-of-the-art optimizer that can handle MILP problems. The code and results are available in \cite{code}.

\subsection{Generation of testing DATs}


\pgfplotstableread{
X Y
0.5000   -0.7455
1.5000   -0.6021
2.5000   -0.3912
3.5000   -0.1160
4.5000    0.2150
5.5000    0.7804
6.5000    1.5589
7.5000    2.0527
}{\mytablenaive}

\pgfplotstableread{
X Y 
0.5000   -0.6922
1.5000   -0.3912
2.5000   -0.1530
3.5000   0.0746
4.5000   0.1938
5.5000   0.3461
6.5000    0.4586
7.5000    0.5969
8.5000    0.6203
9.5000    0.7084
10.5000    0.8045
11.5    0.8599
12.5    0.9005
}{\mytablemilp}

\pgfplotstableread{
X Y 
0.5000   -0.5757
1.5000   -0.3996
2.5000   -0.2321
3.5000   -0.1160
4.5000    0.0280
5.5000    0.0571
6.5000    0.1761
7.5000    0.2372
8.5000    0.3010
9.5000    0.4033
10.5000    0.3785
11.5    0.4073
12.5    0.4319
}{\mytablemnaive}

\pgfplotstableread{
X Y 
0.5000   -0.7270
1.5000   -0.3291
2.5000   -0.0503
3.5000   -0.1232
4.5000    0.2694
5.5000    0.3897
6.5000    0.5098
7.5000    0.6004
8.5000    0.6824
9.5000    0.7804
10.5000    0.7948
11.5    0.8764
12.5    0.9416
}{\mytablemmilp}

To create a testing suite large enough for a meaningful performance comparison, we do the following. As building blocks, we use a selection of DATs from the literature, shown in Table \ref{tab:DATs}. For some, the duration of the BAS are random variables, and we take the expected value for the duration; otherwise we take a random duration from $\{1,2,\ldots,10\}$. We use three methods for combining two DATs $T_1,T_2$ into a larger one (see Figure \ref{fig:combi}):
\begin{enumerate}
\item We take a random BAS $v$ from $T_1$ and consider the modular composition by replacing $v$ in $T_1$ by $T_2$. This represents a larger system, in which one subsystem, represented by $v$ in $T_1$, is given its own DAT for more fine-grained analysis.
\item We introduce a new root node $v$ with a random label, and add edges $(v,\recht{R}_{T_1})$ and $(v,\recht{R}_{T_2})$. This represents a system consisting of two separate subsystems.
\item We introduce a new root node $v$ with a random label, and add edges $(v,\recht{R}_{T_1})$ and $(v,\recht{R}_{T_2})$; we then pick random BAS $b_1$ from $T_1$ and $b_2$ from $T_2$ and identify them (with a new random duration in $\{1,2,\ldots,10\}$). This represents a system consisting of two subsystems that have a shared attack step.
\end{enumerate}

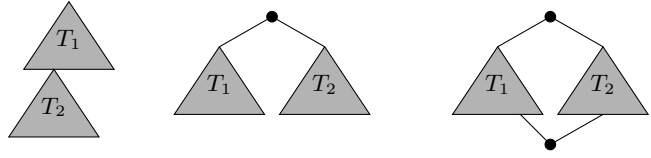
\begin{figure}[t]
\centering
\begin{tikzpicture}
\draw[fill = black!30] (0,0) -- (-0.6,-0.9) -- (0.6,-0.9) -- cycle;
\draw (0,-0.5) node {$T_1$};
\draw[fill = black!30] (-0.2,-0.9) -- (-0.8,-1.8) -- (0.4,-1.8) -- cycle;
\draw (-0.2,-1.4) node {$T_2$};
\draw[fill = black!30] (2,-0.6) -- (1.4,-1.5) -- (2.6,-1.5) -- cycle;
\draw (2,-1.1) node {$T_1$};
\draw[fill = black!30] (3.4,-0.6) -- (2.8,-1.5) -- (4,-1.5) -- cycle;
\draw (3.4,-1.1) node {$T_2$};
\draw (2,-0.6) -- (2.7,-0.2) -- (3.4,-0.6);
\filldraw (2.7,-0.2) circle (2pt);

\draw[fill = black!30] (5.7,-0.6) -- (5.1,-1.5) -- (6.3,-1.5) -- cycle;
\draw (5.7,-1.1) node {$T_1$};
\draw[fill = black!30] (7.1,-0.6) -- (6.5,-1.5) -- (7.7,-1.5) -- cycle;
\draw (7.1,-1.1) node {$T_2$};
\draw (5.7,-0.6) -- (6.4,-0.2) -- (7.1,-0.6);
\filldraw (6.4,-0.2) circle (2pt);
\draw (6,-1.5) -- (6.4,-1.9) -- (7.1,-1.5);
\filldraw (6.4,-1.9) circle (2pt);
\end{tikzpicture}
\caption{The three ways of combining DATs.} \label{fig:combi}
\end{figure}

These are not the only ways by which multiple DATs can be combined; for instance, $T_1$ and $T_2$ could share multiple BAS. We selected these three methods to capture some of the common ways DATs are created by experts. Creating a benchmark suite of large DATs that resemble DATs from industry is an important avenue for further research, but beyond the scope of this paper.

We create two suites of testing DATs by combining the DATs from Table \ref{tab:DATs}. For the first suite, $\mathcal{A}$, we combine DATs using one of the three methods above (drawn randomly) until the result has a given number of nodes. The resulting will have many modules, as $T_1$ is a module under the first method, and both $T_1$ and $T_2$ are modules under the second method. Therefore, we expect the modular approaches to be very fast on the DATs in $\mathcal{A}$. To also study DATs with less modules, we create the second suite, $\mathcal{B}$, by combining DATs using only the third method. Again, one could assign other weights to the three combination methods to obtain yet different testing suites, but $\mathcal{A}$ and $\mathcal{B}$ represent two of the extremes of what DATs can look like.

For a given $n_{\recht{min}}$, we combine DATs randomly drawn from Table \ref{tab:DATs} (either via randomly drawn methods from the 3 above, or by method 3 only) until $|N| \geq n_{\recht{min}}$. We do this 5 times for each $1 \leq n_{\recht{min}} \leq 240$, giving us two testing sets $\mathcal{A},\mathcal{B}$ of 1200 DATs with $8 \leq |N| \leq 262$. On average 26.6\% of the nodes of ATs in $\mathcal{A}$, and 16.5\% of the nodes of ATs in $\mathcal{B}$ are modules. Furthermore 54.2\% of the nodes of ATs in $\mathcal{A}$, and 52.5\% of the nodes of ATs in $\mathcal{B}$ are BAS.

\subsection{Time comparisons}

\begin{table*}[t]
\centering
\renewcommand{\arraystretch}{1}

\begin{tabular}{l|cccc|ccc}
& \multicolumn{4}{c|}{$\mathcal{A}_{\recht{small}}$} & \multicolumn{3}{c}{$\mathcal{A}$} \\
& $\mathtt{MT\text{-}Enum}$ & $\mathtt{MT\text{-}MILP}$ & $\mathtt{MT\text{-}Enum_{Mod}}$ & $\mathtt{MT\text{-}MILP_{Mod}}$ & $\mathtt{MT\text{-}MILP}$ & $\mathtt{MT\text{-}Enum_{Mod}}$ & $\mathtt{MT\text{-}MILP_{Mod}}$\\ \hline
Median time & 1.234 & 0.906 & 1.461 & 1.680 & 1.422 & 2.797 & 3.070 \\
Max time & 10000 & 7.984 & 12.656 & 6.656 & 19.125 & 10000 & 30.469 \\
Failure & 3.71\% & 0\% & 0\% & 0\% & 0\% & 0.08\% & 0\% \\ \hline
& \multicolumn{4}{c|}{$\mathcal{B}_{\recht{small}}$} & \multicolumn{3}{c}{$\mathcal{B}$} \\ \hline
Median time & 1.391 & 0.938 & 1.469 & 1.656 & 1.266 & 3.203 & 2.773 \\
Max time & 10000 & 4.75 & 2326 & 9.484 & 4.75 & 10000 & 9.484 \\
Failure & 3.81\% & 0\%& 0\%& 0\% & 0\% & 3.08\% & 0\%
\end{tabular}

\caption{Summary of the results. All times are in seconds. \emph{Failure} denotes failure to compute within $10^4$ seconds. $\mathcal{A}_{\recht{small}}$ contains $754$ DATs with $\leq 160$ nodes, and $\mathcal{A}$ contains 1200 DATs with $\leq 262$ nodes (including those of $\mathcal{A}_{\recht{small}}$). The sets $\mathcal{B}_{\recht{small}}$ and $\mathcal{B}$ hold the same amount of DATs of the same size; they are designed to contain less modules.} \label{tab:results}
\end{table*}

We measure the computation time of the four algorithms on the testing set; we cap computation time per DAT at $10^4$ seconds. We group the DATs depending on their value of $\lceil |N|/20 \rceil$ and calculate the median per group: these are presented in Figure \ref{fig:lines}. We use the median because it allows us to incorporate the computations that were cancelled after $10^4$ seconds. Since already 21.3\% of the DATs of $\mathcal{A}$, and 13.8\% of DATs of $\mathcal{B}$, with $141 \leq |N| \leq 160$ fail to compute for $\mathtt{{MT}\text{-}{Enum}}$, we do not continue testing this method for larger DATs. The subsetsof $\mathcal{A},\mathcal{B}$ of DATs with $|N| \leq 160$ is called $\mathcal{A}_{\recht{small}},\mathcal{B}_{\recht{small}}$, and consist of $754$ resp. $761$ DATs. The results are also summarized in Table \ref{tab:results}, and pairwise comparisons are presented in Figure \ref{fig:scatter}. 

On the testing set $\mathcal{A}$, we see from Figure \ref{fig:lines}
that $\mathtt{MT\text{-}Enum}$ is by far the slowest method, while $\mathtt{MT\text{-}MILP}$ is the fastest; the two modular approaches are slightly slower than $\mathtt{MT\text{-}MILP}$ and have similar efficiency. While the inefficiency of $\mathtt{MT\text{-}Enum}$ is to be expected, it is surprising that modular analysis for MILP has a net negative effect on computation time. One possible reason is that the Gurobi solver, which we treat as a black box, might incorporate strategies to reduce the MILP problem complexity that are equivalent to modular analysis on the DAT side. At any rate, the enumerative approach clearly shows the advantage of incorporating the modular approach. These results are also reflected in Figure \ref{fig:scatter}(a)--(d).

\begin{figure}[t]
\begin{subfigure}{0.49\linewidth}
\centering
\begin{tikzpicture}
\begin{axis}[  
width = 6.5cm,
height = 4cm,
    legend style={legend pos = north west,legend columns=2},     
    ymin = -1,
    ymax = 3,
    ytick={-2,-1,0,1,2,3},
    yticklabels={$10^{-2}$,$10^{-1}$,$10^0$,$10^1$,$10^2$,$10^3$},
    ylabel near ticks, 
    xmin = 0,
    xmax = 13,
    xtick={2,4,6,8,10,12},  
    xticklabels = {40,80,120,160,200,240},
    nodes near coords align={vertical},  
    ] 
\addplot
  plot[color=blue,mark=triangle*,mark options = {fill=blue}]  table [x=X,y=Y] {\mytablenaive}; \label{pl:naive}
 \addplot
  plot[color=red,mark=square*,mark options = {fill=red}] table [x=X,y=Y] {\mytablemilp}; \label{pl:milp}
\addplot
  plot[color=brown,mark=*,mark options = {fill=brown}] table [x=X,y=Y] {\mytablemnaive}; \label{pl:mnaive}
\addplot
  plot[color=black,mark=diamond*,mark options = {fill=black}] table [x=X,y=Y] {\mytablemmilp}; \label{pl:mmilp}
\end{axis} 
\end{tikzpicture}
\caption{$\mathcal{A}$}
\end{subfigure}
\begin{subfigure}{0.49\linewidth}
\centering
\begin{tikzpicture}
\begin{axis}[  
width = 6.5cm,
height = 4cm,
    legend style={legend pos = north west,legend columns=2},     
    ymin = -1,
    ymax = 3,
    ytick={-2,-1,0,1,2,3},
    yticklabels={$10^{-2}$,$10^{-1}$,$10^0$,$10^1$,$10^2$,$10^3$},
    ylabel near ticks, 
    xmin = 0,
    xmax = 13,
    xtick={2,4,6,8,10,12},  
    xticklabels = {40,80,120,160,200,240},
    nodes near coords align={vertical},  
    ] 
    \addplot
  plot[color=blue,mark=triangle*,mark options = {fill=blue}]  table [x=X,y=Y] {\mytablenaivedag};
 \addplot
  plot[color=red,mark=square*,mark options = {fill=red}] table [x=X,y=Y] {\mytablemilpdag};
\addplot
  plot[color=brown,mark=*,mark options = {fill=brown}] table [x=X,y=Y] {\mytablemnaivedag};
\addplot
  plot[color=black,mark=diamond*,mark options = {fill=black}] table [x=X,y=Y] {\mytablemmilpdag};
\end{axis} 
\end{tikzpicture}
\caption{$\mathcal{B}$}
\end{subfigure}
\caption{Median time (in seconds) of \ref{pl:naive} $\mathtt{MT\text{-}Enum}$, \ref{pl:milp} $\mathtt{MT\text{-}MILP}$, \ref{pl:mnaive} $\mathtt{MT\text{-}Enum_{Mod}}$, \ref{pl:mmilp} $\mathtt{MT\text{-}MILP_{Mod}}$, grouped by the number of nodes $|N|$.}\label{fig:lines}
\end{figure}

\begin{figure}[t]
    \centering
    \begin{subfigure}[b]{0.32\textwidth}
    \centering
		\begin{tikzpicture}
		\tikzstyle{every node}=[font=\small]
		\begin{axis}[
		width = 4.5cm,
		height = 4.5cm,
		xmin = -2,
		xmax = 4,
		ymin = -2,
		ymax = 4,
		ylabel near ticks,
		xtick = {-2,-1,0,1,2,3,4},
		xticklabels = {$10^{-2}$,,$10^{0}$,,$10^2$,,$10^4$},
		ytick = {-2,-1,0,1,2,3,4},
		yticklabels = {$10^{-2}$,,$10^{0}$,,$10^2$,,$10^4$},
		]
		\addplot[only marks,mark=x] table[col sep = comma] {EvM_suc.csv};
		\addplot[only marks,mark=o,color=purple] table[col sep = comma] {EvM_fail.csv};
		\addplot[mark=none, blue] coordinates {(-2,-2) (4,4)};
		\end{axis}
		\end{tikzpicture}
		\caption{$\mathtt{MILP}$, $\mathtt{Enum}$ on $\mathcal{A}_{\recht{small}}$}
\end{subfigure}
    \begin{subfigure}[b]{0.32\textwidth}
    \centering
		\begin{tikzpicture}
		\tikzstyle{every node}=[font=\small]
		\begin{axis}[
		width = 4.5cm,
		height = 4.5cm,
		xmin = -2,
		xmax = 4,
		ymin = -2,
		ymax = 4,
		ylabel near ticks,
		xtick = {-2,-1,0,1,2,3,4},
		xticklabels = {$10^{-2}$,,$10^0$,,$10^2$,,$10^4$},
		ytick = {-2,-1,0,1,2,3,4},
		yticklabels = {$10^{-2}$,,$10^0$,,$10^2$,,$10^4$},
		]
		\addplot[only marks,mark=x] table[col sep = comma] {EvEm_suc.csv};
		\addplot[only marks,mark=o,color=purple] table[col sep = comma] {EvEm_fail.csv};
		\addplot[mark=none, blue] coordinates {(-2,-2) (4,4)};
		\end{axis}
		\end{tikzpicture}
    \caption{$\mathtt{Enum}$, $\mathtt{Enum_{Mod}}$ on $\mathcal{A}_{\recht{small}}$}
    \end{subfigure}
\begin{subfigure}[b]{0.32\textwidth}
    \centering
		\begin{tikzpicture}
		\tikzstyle{every node}=[font=\small]
		\begin{axis}[
		width = 4.5cm,
		height = 4.5cm,
		xmin = -2,
		xmax = 4,
		ymin = -2,
		ymax = 4,
		ylabel near ticks,
		xtick = {-2,-1,0,1,2,3,4},
		xticklabels = {$10^{-2}$,,$10^{0}$,,$10^2$,,$10^4$},
		ytick = {-2,-1,0,1,2,3,4},
		yticklabels = {$10^{-2}$,,$10^{0}$,,$10^2$,,$10^4$},
		]
		\addplot[only marks,mark=x] table[col sep = comma] {MvMm.csv};
		\addplot[mark=none, blue] coordinates {(-2,-2) (4,4)};
		\end{axis}
		\end{tikzpicture}
    \caption{$\mathtt{MILP_{Mod}}$, $\mathtt{MILP}$ on $\mathcal{A}$}
\end{subfigure}
\begin{subfigure}[b]{0.32\textwidth}
    \centering
		\begin{tikzpicture}
		\tikzstyle{every node}=[font=\small]
		\begin{axis}[
		width = 4.5cm,
		height = 4.5cm,
		xmin = -2,
		xmax = 4,
		ymin = -2,
		ymax = 4,
		ylabel near ticks,
		xtick = {-2,-1,0,1,2,3,4},
		xticklabels = {$10^{-2}$,,$10^{0}$,,$10^2$,,$10^4$},
		ytick = {-2,-1,0,1,2,3,4},
		yticklabels = {$10^{-2}$,,$10^{0}$,,$10^2$,,$10^4$},
		]
		\addplot[only marks,mark=x] table[col sep = comma] {EmvMm_suc.csv};
		\addplot[only marks,mark=o,color=purple] table[col sep = comma] {EmvMm_fail.csv};
		\addplot[mark=none, color=blue] coordinates {(-2,-2) (4,4)};
		\end{axis}
		\end{tikzpicture}
    \caption{$\mathtt{Enum_{Mod}}$, $\mathtt{MILP_{Mod}}$ on $\mathcal{A}$}
\end{subfigure}
\begin{subfigure}[b]{0.32\textwidth}
    \centering
		\begin{tikzpicture}
		\tikzstyle{every node}=[font=\small]
		\begin{axis}[
		width = 4.5cm,
		height = 4.5cm,
		xmin = -2,
		xmax = 4,
		ymin = -2,
		ymax = 4,
		ylabel near ticks,
		xtick = {-2,-1,0,1,2,3,4},
		xticklabels = {$10^{-2}$,,$10^{0}$,,$10^2$,,$10^4$},
		ytick = {-2,-1,0,1,2,3,4},
		yticklabels = {$10^{-2}$,,$10^{0}$,,$10^2$,,$10^4$},
		]
		\addplot[only marks,mark=x] table[col sep = comma] {MvMm_DAG.csv};
		\addplot[mark=none, color=blue] coordinates {(-2,-2) (4,4)};
		\end{axis}
		\end{tikzpicture}
    \caption{$\mathtt{MILP_{Mod}}$, $\mathtt{MILP}$ on $\mathcal{B}$} \label{fig:scatter_dag_milp}
\end{subfigure}
\begin{subfigure}[b]{0.32\textwidth}
    \centering
		\begin{tikzpicture}
		\tikzstyle{every node}=[font=\small]
		\begin{axis}[
		width = 4.5cm,
		height = 4.5cm,
		xmin = -2,
		xmax = 4,
		ymin = -2,
		ymax = 4,
		ylabel near ticks,
		xtick = {-2,-1,0,1,2,3,4},
		xticklabels = {$10^{-2}$,,$10^{0}$,,$10^2$,,$10^4$},
		ytick = {-2,-1,0,1,2,3,4},
		yticklabels = {$10^{-2}$,,$10^{0}$,,$10^2$,,$10^4$},
		]
		\addplot[only marks,mark=x] table[col sep = comma] {EmvMm_DAG_suc.csv};
		\addplot[only marks,mark=o,color=purple] table[col sep = comma] {EmvMm_DAG_fail.csv};
		\addplot[mark=none, color=blue] coordinates {(-2,-2) (4,4)};
		\end{axis}
		\end{tikzpicture}
    \caption{$\mathtt{Enum_{Mod}}$, $\mathtt{MILP_{Mod}}$ on $\mathcal{B}$} \label{fig:scatter_dag_enum}
\end{subfigure}
\caption{Pairwise computation time comparisons of the four algorithms. The first algorithm is the vertical axis while the second is the horizontal axis. Each mark is a DAT; purple circles are computations aborted for exceeding $10^4$ seconds.} \label{fig:scatter}
\end{figure}
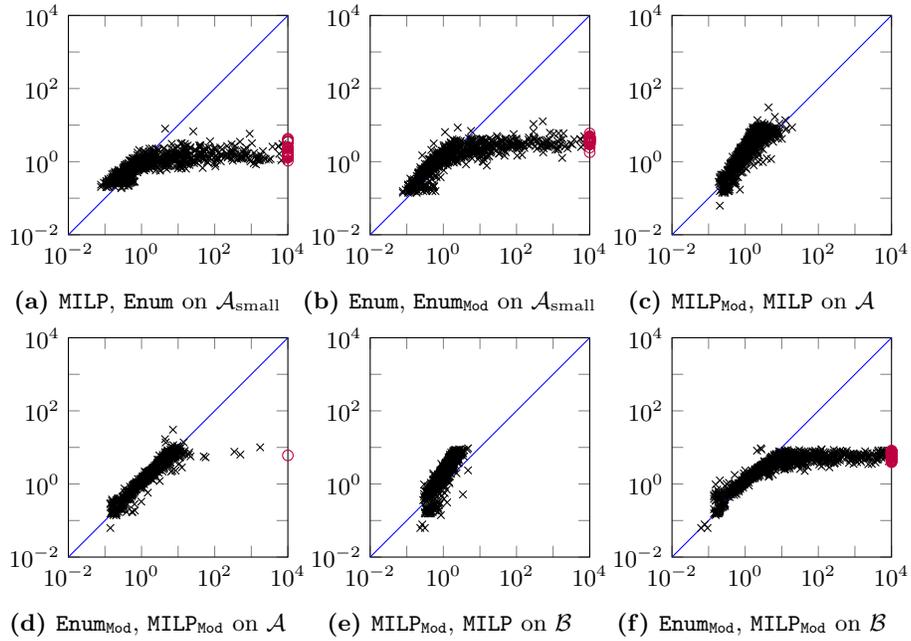

Interestingly, the difference in median computation time between $\mathtt{MT\text{-}Enum}$ and $\mathtt{MT\text{-}MILP}$ disappears when considering the modular versions of these algorithms, although the worst-case behaviour of $\mathtt{MT\text{-}Enum_{Mod}}$ is considerably worse than that of $\mathtt{MT\text{-}MILP_{Mod}}$ (see Table \ref{tab:results}). We hypothesize that this is due to the fact that the DATs of $\mathcal{A}$ contain many modules. As a result, the `indecomposable' sub-DATs on which the algorithms $\mathtt{MT\text{-}Enum}$ and $\mathtt{MT\text{-}MILP}$ are called will typically be small. Since the difference in computation time between these algorithms only appears for larger DATs, we do not see it in these experiments. 

For testing set $\mathcal{B}$, we again see that $\mathtt{MT\text{-}Enum}$ is by far the slowest. Furthermore, for larger DATs $\mathtt{MT\text{-}MILP_{Mod}}$ outpaces $\mathtt{MT\text{-}Enum_{Mod}}$ considerably; see also Figure \ref{fig:scatter}(f). This shows that also in a modular setting the MILP approach significantly speeds up calculations for large enough DATs. This is to be expected from our results on set $\mathcal{A}$ as for larger DATs the `indecomposable' subDATs on which $\mathtt{MT\text{-}MILP}$ is invoked will be larger as well. Interestingly, on this dataset $\mathtt{MT\text{-}MILP}$ is slightly faster than $\mathtt{MT\text{-}MILP_{Mod}}$, as can also be seen from Figure \ref{fig:scatter}(e). This might be due to the fact that on wide DATs, the MILP methods of Gurobi are more efficient at splitting up DATs into modules than our Matlab implementation of the  modular decomposition algorithm. A detailed study into this difference in performance would entail a comprehensive analysis into Gurobi's Matlab implementation, which is beyond the scope of this paper.

Taking $\mathcal{A}$ and $\mathcal{B}$ together, we can conclude that both the MILP approach and modular analysis create a large decrease in computation time. While these methods are slightly slower for small DATs, computation time for such DATs only takes a few seconds anyway. By contrast, for larger DATs the difference in computation time can go up to a factor $10^3$. For DATs with large modules, $\mathtt{MT\text{-}Enum_{Mod}}$ loses out against $\mathtt{MT\text{-}MILP}$ and $\mathtt{MT\text{-}MILP_{Mod}}$, which behave similarly. 

\section{Conclusion and discussion}

This paper introduced two novel tools to calculate \emph{min time} for DATs. First, we introduced a novel MILP-based approach that finds \emph{min time} by phrasing it as an optimization problem. Second, we show how modular analysis can be used to reduce the computation time of any \emph{min time} calculation algorithm. In the experiments, we compared these to the enumerative method. The experiments show that for large DATs both MILP and modular analysis can have a big impact on computation time. In particular, the MILP approach is consistently fast on any input DAT, making it a reliable tool for quantitative DAT analysis in practice.

There are several directions in which this work can be expanded. First, a benchmark suite of DATs is needed. For this it is important to find out what sizes and properties are typical for DATs used in industry, even if industry DATs themselves may not be published due to confidentiality reasons.

Second, modular analysis can also be used for other metrics, as has been done for fault trees \cite{Reay2002,Ruijters2015}. Since modular analysis is a very general idea, a good approach would be to develop an axiomatization of metrics that can be handled via modular analysis, so that the method can be applied to a large set of metrics at once. Such a result is probably not hard to prove for metrics that are defined bottom-up as in \cite{Kordy2018}; the challenge lies in metrics that are defined directly from the semantics as in \cite{Budde2021}.

Third, our MILP approach can be combined with a Monte Carlo approach in a stochastic setting where the precise BAS values are unknown. A more thorough investigation can explore what guarantees such simulations can give for \MinTime. As Monte Carlo methods involve sampling a large sample, performance of the \MinTime calculation algorithm is important in such a study.

\bibliographystyle{splncs04}
\bibliography{mybibliography}

\appendix

\section{Proofs of Theorem \ref{thm:ta} and Lemma \ref{lem:maxm}} \label{app:thmta}

To prove Theorem \ref{thm:ta} we need some auxiliary results. We start with a stricter definition of time assignments.

\begin{definition}
A time assignment $f$ is called \emph{exact} if in Definition \ref{def:ta}, equality holds in 2) and 3) for all $v$, and equality holds in 4a) if the prerequisite conditions for 4b) and 4c) are not satisfied. The set of exact time assignments is denoted $\mathcal{F}_{\recht{E},T}$.
\end{definition}

The advantage of exact time assignments is that they are easier to reason about. The following lemma shows that restricting ourselves to exact time assignments does not affect the minimum.

\begin{lemma} \label{lem:ta0} $\min_{f \in \mathcal{F}_T} f_{\recht{R}_T} = \min_{f \in \mathcal{F}_{\recht{E},T}} f_{\recht{R}_T}$.
\end{lemma}

\begin{proof}
Let $f \in \mathcal{F}_T$, and define $f' \in \mathcal{F}_{\recht{E},T}$ by $f'_{a} = f_{a}$ for $a \in N_{\mathtt{BAS}}$, and having $f'$ satisfy equality in 2) and 3) for all $v$, and by satisfying equality in 4a) if the prerequisite conditions for 4b) and 4c) are not satisfied. By induction it is straightforward to show that $f'_{\recht{R}_T} \leq f_{\recht{R}_T}$, which proves the lemma.
\end{proof}

The following lemma has a straightforward proof by induction which is therefore omitted.

\begin{lemma} \label{lem:ta1}
Let $f$ be an exact time assignment with $f_{\recht{R}_T} < \infty$. Then $f_{v} \in \{f_{a}: a \in B_{v}\} \cup \{\infty\}$ for all $v \in N$. \qed
\end{lemma}

We need two more lemmas as ingredients for the proof of Theorem \ref{thm:ta}.

\begin{lemma} \label{lem:ta2}
Let $f$ be an exact time assignment. Then there exists a successful attack $\mathcal{O}$ with $\recht{t}(\mathcal{O}) \leq f_{\recht{R}_T}$.
\end{lemma}

\begin{proof}
Let $C \in (0,\infty)$. Define the attack $\mathcal{O}^C_f = (A^C_f,\prec^C_f)$ by
\begin{align*}
A^C_f &= \{a \in N_{\mathtt{BAS}}:f_{a} \leq C\},\\
\prec^C_f &= \{(a,a') \in (A^C_f)^2: f_{a} \leq f_{a'}-d(a')\}.
\end{align*}
We prove by induction on $T$ that for any $f$ and $C$ one has that $\mathcal{O}^C_f$ reaches a node $v$ when $f_{v} \leq C$. For convenience we assume that $T$ is binary; this does not substantially alter the proof but makes the notation easier.
\begin{itemize}
    \item Suppose $a \in N_{\mathtt{BAS}}$, and suppose $f_{a} \leq C$. Then $a \in A^C_f$ and $\mathcal{O}^C_f$ is successful.
    \item Suppose $v = \mathtt{OR}(v_1,v_2)$ and $f_{v} \leq C$. Since $f_{v} = \min\{f_{v_1},f_{v_2}\}$, we may assume WLOG that $f_{v_1} \leq C$. By the induction hypothesis the attack $\mathcal{O}^C_{f}$ reaches $v_1$. It follows that $\mathcal{O}^C_f$ reaches $v_1$, and so it reaches $v$ too.
    \item The case that $v= \mathtt{AND}(v_1,v_2)$ is analogous to the $\mathtt{OR}$-case.
    \item Suppose that $v = \mathtt{SAND}(v_1,v_2)$ and $f_{v} \leq C$; then \begin{equation}
    f_{v_1}\leq f_{v_2} = f_{v} \leq C <\infty    
    \end{equation} 
    and $f_{a} \leq f_{a'}-d(a')$ for all $a \in B_{v_1}$ and $a' \in B_{v_2}$ for those $a$ with $f_{a} < \infty$. By the induction hypothesis, one proves as in the $\mathtt{OR}$-case that $\mathcal{O}^C_f$ reaches both $v_1$ and $v_2$. Furthermore, by definition of $\prec^C_f$ one has $a \prec^C_f a'$ for all $a \in A^C_f \cap B_{v_1}$ and $a' \in A^C_f \cap B_{v_2}$. We conclude that $\mathcal{O}^C_f$ reaches $v$.
\end{itemize}
Now take $C = f_{\recht{R}_T} < \infty$; the resulting attack $\mathcal{O}^{f_{\recht{R}_T}}_f$ is successful. Lemma \ref{lem:ta1} then tells us that
\begin{equation*}
f_{\recht{R}_T} \in \left\{f_{a}: a \in A^{f_{\recht{R}_T}}_f\right\}.
\end{equation*}
Since the maximum of the RHS is at most equal to $f_{\recht{R}_T}$ we find
\begin{equation} \label{eq:pfta1}
  f_{\recht{R}_T} = \max_{a \in A^{f_{\recht{R}_T}}_f} f_{a}.
\end{equation}
On the other hand, by induction it can be shown straightforwardly that for every chain $C$ in $\mathcal{O}^{f_{\recht{R}_T}}_f$ one has
\begin{equation*}
    f_{\max C} \geq \sum_{a \in C} d(a).
\end{equation*}
Combining this with \eqref{eq:pfta1} we find $f_{\recht{R}_T} \geq \recht{t}(\mathcal{O}^{f_{\recht{R}_T}}_f)$.
\end{proof}

\begin{lemma} \label{lem:ta3}
Let $\mathcal{O}$ be a successful attack. Then there exists an exact time assignment $f$ such that $f_{\recht{R}_T} \leq \recht{t}(\mathcal{O})$.
\end{lemma}

\begin{proof}
Let $\mathcal{O} = (A,\prec)$. Define an exact time assignment $f_{\mathcal{O}}$ as follows:
\begin{itemize}
    \item If $a \in N_{\mathtt{BAS}} \setminus A$, take $f_{\mathcal{O},a} = \infty$;
    \item If $a \in A$, define $f_{\mathcal{O}}$ recursively by
    \begin{equation*}
        f_{\mathcal{O},a} = d(a)+\max_{a' \prec a} f_{\mathcal{O},a'},
    \end{equation*}
    where the maximum over the empty set equals $0$;
    \item If $\gamma(v) \in \{\mathtt{OR},\mathtt{AND}\}$, define $f_{\mathcal{O},v}$ by taking equality in 2) and 3) of Definition \ref{def:ta}; if $\gamma(v) = \mathtt{SAND}$, take equality in 4a) if the conditions for 4b) and 4c) are not satisfied, and otherwise take $f_{\mathcal{O},v} = \infty$.
\end{itemize}
We prove by induction on $T$ that $f_{\mathcal{O},v}< \infty$ if $\mathcal{O}$ reaches $v$. For convenience we again assume that $T$ is binary.
\begin{itemize}
    \item When $v \in N_{\mathtt{BAS}}$ it is clear from the definition.
    \item Suppose $v = \mathtt{OR}(v_1,v_2)$. WLOG $\mathcal{O}$ reaches $v_1$, so by the induction hypothesis $f_{\mathcal{O},v_1} < \infty$. As $f_{\mathcal{O},v} \leq f_{\mathcal{O},v_1}$ the induction hypothesis holds for $v$.
    \item The case $v = \mathtt{AND}(v_1,v_2)$ is analogous to the $\mathtt{OR}$-case.
    \item Suppose $v = \mathtt{SAND}(v_1,v_2)$. Since $v$ is reached, this means that $v_1$ and $v_2$ are reached and that $a \prec a'$ for all $a \in B_{v_1} \cap A$, $a' \in B_{v_2} \cap A$. The first statement implies, by the induction hypothesis, that $f_{v_1},f_{v_2} < \infty$. The second statement, together with the definition of $f_{\mathcal{O}}$ on BAS, implies that $f_{a} \leq f_{a'}-d(a')$ for $a,a'$ as above, which means that the prerequisite condition for 4c) never holds. It follows that $f_{\mathcal{O},v} < \infty$.
\end{itemize}
Since $\mathcal{O}$ is successful one has $f_{\mathcal{O},\recht{R}_T} < \infty$, so by Lemma \ref{lem:ta1} we know that 
\begin{equation*}
f_{\mathcal{O},\recht{R}_T} \in \{f_{a}\colon a \in A\}.
\end{equation*}
However, by induction on the strict poset $\mathcal{O}$ one can prove that $\recht{t}(\mathcal{O}) = \max_{a \in A} f_{\mathcal{O},a}$. It follows that $\recht{t}(\mathcal{O}) \geq f_{\mathcal{O},\recht{R}_T}$, as was to be shown.
\end{proof}

\begin{proof}[Proof of Theorem \ref{thm:ta}]
By Lemma \ref{lem:ta0} it is enough to only consider exact time assignments. From Lemma \ref{lem:ta2} it follows that
\begin{equation*}
\min_{\mathcal{O} \in \mathcal{S}_T} \recht{t}(\mathcal{O}) \leq \min_{f \in \mathcal{F}_{\recht{E},T}} f_{\recht{R}_T}.
\end{equation*}
On the other hand, Lemma \ref{lem:ta3} tells us that
\begin{equation*}
\min_{\mathcal{O} \in \mathcal{S}_T} \recht{t}(\mathcal{O}) \geq \min_{f \in \mathcal{F}_{\recht{E},T}} f_{\recht{R}_T}.
\end{equation*}
Since the LHS of these two inequalities is equal to $\recht{mt}(T)$, these two inequalities together prove the theorem.
\end{proof}

\begin{proof}[Proof of Lemma \ref{lem:maxm}]
If $\mathcal{S}_T = \varnothing$, then the optimal $f$ has $f_{\recht{R}_T} = \infty$; hence the constant time assignment $f \equiv \infty$ also minimizes \eqref{eq:opt}.

Suppose $\mathcal{S}_T \neq \varnothing$, and let $f$ be a time assignment minimizing \eqref{eq:opt}. Since $\mathcal{S}_T \neq \varnothing$ one has $f_{\recht{R}_T} = \recht{mt}(T) < \infty$. Let $\mathcal{O}$ be as in Lemma \ref{lem:ta2}, and let $f' = f_{\mathcal{O}}$ be as in the proof of Lemma \ref{lem:ta3}. Then $f'_{\recht{R}_T} \leq f_{\recht{R}_T} = \recht{mt}(T)$, but since $f_{\recht{R}_T}$ is minimal this is an equality. Furthermore, is straightforward to prove by induction that $f'_v \leq M-1$ for all $v$ from the definition of $f_{\mathcal{O}}$.
\end{proof}

\section{Proof of Theorem \ref{thm:prune}} \label{app:prune}

To prove this theorem, we first prove two auxiliary lemmas. Throughout this section, we write $d_v := d|_{B_v}$ for convenience.

\begin{lemma} \label{lem:prune1}
Let $\mathcal{O}^v$ be a successful attack on $T^v$. then there is a successful attack $\mathcal{O}$ on $T$ with $\recht{t}(\mathcal{O},d) = \recht{t}(\mathcal{O}^v,d^v)$.
\end{lemma}

\begin{proof}
Let $\mathcal{O}^v = (A^v,\prec^v)$. If $\tilde{v} \notin A$, then $\mathcal{O}^v$ is also a successful attack on $T$. If we take $\mathcal{O} = \mathcal{O}^v$, then $\recht{t}(\mathcal{O},d) = \recht{t}(\mathcal{O}^v,d^v)$. Now suppose $\tilde{v} \in A$. Let $\mathcal{O}_{v,-} = (A_{v,-},\prec_{v,-})$ be a successful attack on $T_v$ that satisfies $\recht{t}(\mathcal{O}_{v,-},d_v) = \recht{mt}(T_v,d_v)$. Define an attack $\mathcal{O} = (A,\prec)$ on $T$ by $A = A^v\setminus\{\tilde{v}\} \cup A_{v,-}$, and $a \prec a'$ if and only if one of the following holds:
\begin{itemize}
    \item $a,a' \in A^v\setminus\{\tilde{v}\}$ and $a \prec^v a'$;
    \item $a,a' \in A_{v,-}$ and $a \prec_{v,-} a'$;
    \item $a \in A^v\setminus\{\tilde{v}\}$, $a' \in A_{v,-}$, and $a \prec^v \tilde{v}$;
    \item $a \in A_{v,-}$, $a' \in A^v \setminus \{\tilde{v}\}$ and $\tilde{v} \prec^v a'$.
\end{itemize}
Intuitively, $\mathcal{O}$ is the order obtained by taking $\mathcal{O}^v$ and replacing $\tilde{v}$ with the attack $\mathcal{O}_{v,-}$. Since $v$ is a module, $A^v\setminus \{\tilde{v}\}$ and $A_{v,-}$ are disjoint, and $\mathcal{O}$ is an attack. Since it activates $v$ and $\mathcal{O}^v$ is a sucessful attack on $T^v$, the attack $\mathcal{O}'$ is successful on $T$. Furthermore, the maximal chains $C$ of $\mathcal{O}$ are obtained by taking a maximal chain $C^v$ of $\mathcal{O}^v$ and replacing a possible instance of $\tilde{v}$ with a maximal chain $C_{v,-}$ of $\mathcal{O}_{v,-}$. From the definition of $d^v_{\tilde{v}}$ it follows that if $\tilde{v} \notin C^v$ one has
\begin{equation} \label{eq:maxchain1}
    \sum_{a \in C} d(a) = \sum_{a \in C^v} d^v(a),
\end{equation}
and if $\tilde{v} \in C^v$ one has

\begin{align}
    \sum_{a \in C} d(a) &= \sum_{a \in C^v} d^v(a)+\sum_{a \in C_{v,-}} d(a) - d^v(\tilde{v}) \nonumber \\
    &\geq \sum_{a \in C^v} d^v(a), \label{eq:maxchain2}
\end{align}
with equality if and only if $C_{v,-}$ satisfies $\recht{mt}(T_v,d_v) = \sum_{a \in C_{v,-}} d(a)$. Taking the maximum over all $C^v$ in both \eqref{eq:maxchain1} and \eqref{eq:maxchain2}, we find that
\begin{align*}
\recht{t}(\mathcal{O},d) &= \max_{\substack{C \textrm{ max. chain}\\ \textrm{in } \mathcal{O}}} \sum_{a \in C} d(a)\\
&= \max_{\substack{C^v \textrm{ max. chain}\\ \textrm{in } \mathcal{O}^v}} \sum_{a \in C^v} d^v_a\\
&= \recht{t}(\mathcal{O}^v,d^v).
\end{align*}
\end{proof}

\begin{lemma} \label{lem:prune2}
Let $\mathcal{O}$ be a successful attack on $T$. Then there exists a successful attack $\mathcal{O}^v$ on $T^v$ with $\recht{t}(\mathcal{O},d) \geq \recht{t}(\mathcal{O}^v,d^v)$.
\end{lemma}

\begin{proof}
Let $\mathcal{O} = (A,\prec)$. If $\mathcal{O}$ does not reach $v$, define $\mathcal{O}^v = (A^v,\prec^v)$ by
\begin{align*}
A^v &= A \setminus B_{v},\\
\prec^v &= \prec|_{A^v}.
\end{align*}
Then $\mathcal{O}^v$ is a successful attack on both $T^v$ and $T$. Since $\mathcal{O}^v \leq \mathcal{O}$ one has $\recht{t}(\mathcal{O}^v,d^v) = \recht{t}(\mathcal{O}^v,d) \leq \recht{t}(\mathcal{O},d)$. Now suppose $\mathcal{O}$ reaches $v$; then $\mathcal{O}|_{B_{v}}$ is a successful attack on $T_v$. Define another attack $\mathcal{O}_1 = (A_1,\prec_1)$ on $T$ by $A_1 = A \setminus A$ and $a \prec_1 a'$ if and only if one of the following holds:
\begin{itemize}
    \item $a,a' \in A \setminus B_{v}$ and $a \prec a'$;
    \item $a,a' \in B_{v}$ and $a \prec a'$;
    \item $a \in A \setminus B_{v}$, $a' \in B_{v}$ and $a \prec b$ for all $b \in B_{v} \cap A$;
    \item $a \in B_{v}$, $a' \in A \setminus B_{v}$ and $b \prec a'$ for all $b \in B_{v} \cap A$.
\end{itemize}
$\mathcal{O}_1$ is the attack obtained from $\mathcal{O}$ by removing all relations between elements of $A \setminus B_{v}$ and elements of $B_{v}$ that are not shared with all elements of $B_{v}$. Since $v$ is a module, the constraints on $\prec$ regarding relations between elements of $B_{v}$ and elements of $A \setminus B_{v}$ in Definition \ref{def:success} are the same for all elements of $B_v$. Hence we find that $\mathcal{O}_1$ is successful because $\mathcal{O}$ is. Furthermore, $\mathcal{O}_1 \leq \mathcal{O}$, so $\recht{t}(\mathcal{O}_1,d) \leq \recht{t}(\mathcal{O},d)$.

Define an attack $\mathcal{O}^v = (A^v,\prec^v)$ on $T^v$ by $A^v = A_1 \setminus B_{v} \cup \{\tilde{v}\}$ and $a \prec^v a'$ if and only if one of the following holds:
\begin{itemize}
    \item $a,a' \in A_1 \setminus B_{v}$ and $a \prec_1 a'$;
    \item $a \in A_1 \setminus B_{v}$, $a' = \tilde{v}$, and $a \prec_1 b$ for all $b \in A_1 \cap B_{v}$;
    \item $a = \tilde{v}$, $a' \in A_1 \setminus B_{v}$, and $b \prec_1 a'$ for all $b \in A_1 \cap B_{v}$.
\end{itemize}
Since $\mathcal{O}_1$ is successful on $T$ we find that $\mathcal{O}^v$ is successful on $T^v$. Furthermore, the set of maximal chains of $\mathcal{O}_1$ is obtained by taking a maximal chain of $\mathcal{O}^v$ and replacing $\tilde{v}$ (if it occurs) with a maximal chain of $\mathcal{O}|_{B_{v}}$. It follows that 
\begin{align*}
\recht{t}(\mathcal{O}_1,d) &= \max(M_1,M_2-d^v(\tilde{v})+\recht{t}(\mathcal{O}|_{B_{v}},d_v))
\end{align*}
where
\begin{align*}
M_1 &= \max_{\substack{C^v \text{ max. chain of } \mathcal{O}^v\colon\\ \tilde{v} \notin C^v}} \sum_{a \in C^v} d(a),\\
M_2 &= \max_{\substack{C^v \text{ max. chain of } \mathcal{O}^v\colon\\ \tilde{v} \in C^v}} \sum_{a \in C^v} d(a).
\end{align*}
By definition of $d^v(\tilde{v})$ one has $\recht{t}(\mathcal{O}|_{B_{v}},d_v) \geq d^v_{\tilde{v}}$. Since $\recht{t}(\mathcal{O}^v,d^v) = \max\{M_1,M_2\}$, it follows that $\recht{t}(\mathcal{O}_1,d) \geq \recht{t}(\mathcal{O}^v,d^v)$. Since we already know that $\recht{t}(\mathcal{O},d) \geq \recht{t}(\mathcal{O}_1,d)$, we get $\recht{t}(\mathcal{O},d) \geq \recht{t}(\mathcal{O}^v,d^v)$.
\end{proof}

\begin{proof}[Proof of Theorem \ref{thm:prune}]
By Lemma \ref{lem:prune1} one has $\recht{mt}(T,d) \leq \recht{mt}(T^v,d^v)$, and by Lemma \ref{lem:prune2} one has $\recht{mt}(T,d) \geq \recht{mt}(T^v,d^v)$.
\end{proof}

\end{document}